\documentclass[aps,english,prb,floatfix,amsmath,superscriptaddress,tightenlines,twocolumn]{revtex4-2}

\usepackage{amsmath}
\usepackage{amssymb}
\usepackage{amsthm}
\usepackage{amsfonts}

\usepackage{graphicx,color,framed}
\usepackage{hyperref}
\usepackage{times}
\usepackage{enumerate}
\usepackage{lipsum}
\usepackage{slashed}
\usepackage{url}
\usepackage{comment}
\usepackage{bbm}
\usepackage{nicematrix}

\usepackage[caption=false]{subfig}

\usepackage{tikz}
\usetikzlibrary{shapes.geometric}
\usetikzlibrary{arrows.meta,positioning}
\usetikzlibrary{decorations.markings}
\usetikzlibrary{calc}
\usetikzlibrary{decorations.pathreplacing}

\hypersetup{
    colorlinks=true, 
    linktoc=all,     
    linkcolor=blue,  
}

\newcommand{\<}{\langle}
\renewcommand{\>}{\rangle}
\newcommand{\Tr}{\text{Tr}}
\newcommand{\id}{\mathbbm{1}}

\newcommand{\rdet}[1]{\det\nolimits_{#1}}
\newcommand{\imp}{P\mathcal{H}}
\newcommand{\norm}[1]{\left|\left|#1\right|\right|}
\newcommand{\wmn}{\omega_{m,n}}
\newcommand{\rp}[1]{\Tilde{#1}}

\newcommand{\ra}{\rp{\mathcal{A}}}
\newcommand{\rb}{\rp{\mathcal{B}}}
\newcommand{\rc}{\rp{\mathcal{C}}}

\newcommand{\ide}{\id_\epsilon}
\newcommand{\rce}{\rc_\epsilon}
\newcommand{\rae}{\ra_\epsilon}
\newcommand{\uabc}{U_{\mathcal{A}\mathcal{B}\mathcal{C}}}
\newcommand{\mabc}{M_{\mathcal{A}\mathcal{B}\mathcal{C}}}
\newcommand{\abc}{\mathcal{A}\mathcal{B}\mathcal{C}}
\newcommand{\bpm}{\begin{pmatrix}}
\newcommand{\epm}{\end{pmatrix}}
\newcommand{\wab}{\omega_{\alpha,\beta}}

\begin{document}

\title{R\'enyi-like entanglement probe of the chiral central charge}


\author{Julian Gass}
\author{Michael Levin}
\affiliation{Leinweber Institute for Theoretical Physics, University of Chicago, Chicago, Illinois 60637,  USA}

\begin{abstract}
We propose a ground state entanglement probe for gapped, two-dimensional quantum many-body systems that involves taking powers of reduced density matrices in a particular geometric configuration. This quantity, which we denote by $\omega_{\alpha,\beta}$, is parameterized by two positive real numbers $\alpha, \beta$, and can be seen as a ``R\'enyi-like" generalization of the modular commutator -- another entanglement probe proposed as a way to compute the chiral central charge from a bulk wave function. We obtain analytic expressions for $\omega_{\alpha,\beta}$ for gapped ground states of non-interacting fermion Hamiltonians as well as ground states of string-net models. In both cases, we find that $\omega_{\alpha,\beta}$ takes a universal value related to the chiral central charge. For integer values of $\alpha$ and $\beta$, our quantity $\omega_{\alpha,\beta}$ can be expressed as an expectation value of permutation operators acting on an appropriate replica system, providing a natural route to measuring $\omega_{\alpha,\beta}$ in numerical simulations and experiments.
\end{abstract}

\maketitle

\section{Introduction}
The chiral central charge $c_-$ is a rational-valued topological invariant that characterizes two-dimensional (2D) gapped quantum many-body systems. The standard definition of $c_-$ is in terms of the thermal Hall conductance at low temperatures. Specifically, $c_-$ is defined as the dimensionless ratio $c_- = \frac{\kappa_H(T)}{\pi^2 k_B^2 T/3 h}$, where $\kappa_H(T)$ is the thermal Hall conductance at temperature $T$, and $T$ is assumed to be much smaller than the bulk gap.~\cite{KaneFisher_1997, cappelli2002thermal, Kitaev_2006} 

Although it may not be obvious from the above definition, $c_-$ is believed to depend only on the bulk ground state.~\cite{Kitaev_2006, KapustinSpodyneiko_2019} Therefore, it should be possible in principle to compute $c_-$ directly from a ground state wave function on an infinite 2D plane. 

A concrete proposal for such a bulk formula for $c_-$ was made by Refs.~\cite{Kim_2022one, Kim_2022two}. This proposal is based on a quantity known as the ``modular commutator.'' Let $\rho$ be a density operator for a gapped ground state defined on an infinite 2D lattice and let $A, B, C$ be three non-overlapping subsets of the lattice, arranged as in Fig.~\ref{fig:ABC_geometry}. The modular commutator $J$ is defined by
\begin{align}
\label{eq:modularcommutator}
J &=i\<[\ln \rho_{AB},\ln \rho_{BC}]\>,
\end{align}
where $\rho_{R}$ is the reduced density operator on region $R$, and $\<O\>=\Tr (O\rho)$. Refs.~\cite{Kim_2022one, Kim_2022two} conjectured that in the limit of large $A, B, C$, the modular commutator approaches a universal value proportional to the chiral central charge of the ground state:
\begin{align}
\label{eq:CCC}
J =\frac{\pi}{3}c_-.
\end{align}
This conjecture is supported by several pieces of evidence including (i) a derivation of (\ref{eq:CCC})~\cite{Kim_2022one, Kim_2022two} in the case where $\rho$ satisfies the entanglement bootstrap axioms~\cite{bootstrap}, (ii) a proof in the case where $\rho$ is a gapped ground state of a non-interacting fermion Hamiltonian~\cite{Fan_2023}, and (iii) numerical results for a Laughlin-like ground state~\cite{Kim_2022two} (see also Refs.~\cite{Zou_mod_comm_CFT_2022, Fan_mod_comm_CFT_2022} for related CFT-based calculations). 

In this paper we propose a R\'enyi-like generalization of the modular commutator. Our R\'enyi modular commutator, $\wab$, is parameterized by two positive real numbers $\alpha, \beta$ and is defined as the phase of $\<\rho_{AB}^\alpha\rho_{BC}^\beta\>$:\footnote{If $\<\rho_{AB}^\alpha\rho_{BC}^\beta\> = 0$ then $\wab$ is not well-defined. However, we expect that generically $\<\rho_{AB}^\alpha\rho_{BC}^\beta\> \neq 0$.}
\begin{equation}
\label{wmn_definition}
    \wab=\frac{\<\rho_{AB}^\alpha\rho_{BC}^\beta\>}{|\<\rho_{AB}^\alpha\rho_{BC}^\beta\>|}.
\end{equation}
One can check that $\wab$ reduces to the original modular commutator (\ref{eq:modularcommutator}) in the limit $\alpha, \beta \rightarrow 0$: in particular, if we define $J_{\alpha,\beta} \equiv \frac{2i}{\alpha\beta} \ln \wab$, then
\begin{equation}
    \lim_{\alpha,\beta\rightarrow 0} J_{\alpha,\beta} =i\<[\ln\rho_{AB},\ln\rho_{BC}]\> = J.
\label{reductmodcomm}
\end{equation}

\begin{figure}
    \centering
    \includegraphics[width=0.4\columnwidth]{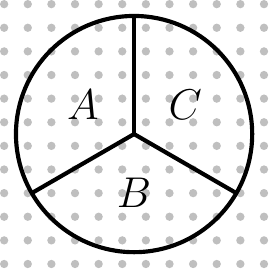}
    \caption{Geometry used to compute the modular commutator $J$ and its R\'enyi generalization $\wab$.}
    \label{fig:ABC_geometry}
\end{figure}

To investigate this R\'enyi modular commutator, we compute $\wab$ for two classes of ground states: (i) gapped ground states of non-interacting fermion Hamiltonians and (ii) ground states of string-net models~\cite{Levinwen2005, hong2009symmetrization, Hahn2020, LinLevinBurnell2021} -- exactly solvable spin models that realize a large class of (interacting) topological phases. In both cases, we find that $\wab$ takes the following universal value in the limit of large $A, B, C$: 
\begin{equation}
\label{wmn_conjecture}
\begin{gathered}
    \wab=\exp\left\{-\frac{\pi i}{12}q(\alpha,\beta) c_-\right\},\\
    q(\alpha,\beta)=\frac{\alpha}{\alpha+1}+\frac{\beta}{\beta+1}-\frac{\alpha+\beta}{\alpha+\beta+1}.
\end{gathered}
\end{equation}
More precisely, in the non-interacting fermion case, we show that $\wab = \exp[-\frac{\pi i}{24}q(\alpha,\beta) \nu(P)]$ where $\nu(P)$ is the Chern number of the corresponding spectral projector $P$ in the Majorana basis \footnote{This result holds up to corrections that are exponentially small in the size of the $A,B,C$ regions for integer $\alpha$ and $\beta$. On the other hand, for non-integer $\alpha$ and $\beta$, our error bound scales \emph{algebraically} with the size of $A, B, C$ -- see the end of Sec.~\ref{sec:computing_wmn} for more details.}. This result can be rephrased as Eq.~(\ref{wmn_conjecture}) since $c_- = \frac{\nu(P)}{2}$. Likewise, for ground states of string-net models, we show that $\wab = 1$, which is again consistent with (\ref{wmn_conjecture}), since $c_- = 0$ for these states, as reviewed below.

We also find several other pieces of evidence for Eq.~(\ref{wmn_conjecture}): first, one can check that (\ref{wmn_conjecture}) reduces to the original modular commutator conjecture (\ref{eq:CCC}) in the limit $\alpha, \beta \rightarrow 0$. In addition, $\wab$ has several general properties that are consistent with (\ref{wmn_conjecture}), including being multiplicative under stacking, odd under time reversal symmetry, and taking the trivial value $\wab = 1$ for any tripartite pure state that is supported on only three of the four regions $A, B, C, D \equiv(ABC)^c$.   

Despite this evidence, the above formula for $\wab$ (\ref{wmn_conjecture}) does not hold in complete generality: there exist (fine-tuned) examples of gapped Hamiltonians whose ground states violate (\ref{wmn_conjecture})\footnote{One example of this type is the model discussed in Ref.~\cite{Gass_2024} with a spurious modular commutator. Not surprisingly, this model also has spurious values of $\wab$.}, just as there are examples that violate (\ref{eq:CCC})~\cite{Gass_2024}, and examples that violate the topological entanglement entropy conjecture~\cite{zou_2016,Williamson_2019}. In these examples, $\wab$ takes non-universal ``spurious'' values which are sensitive to the specific choice of $A, B, C$ and which are unrelated to the chiral central charge. Moreover, we expect that there exist (fine-tuned) examples of gapped many-body states where $\<\rho_{AB}^\alpha\rho_{BC}^\beta\>=0$ so that $\wab$ is not even well-defined. 
Therefore, we fall short of claiming that $\wab$ is a true invariant. 

Instead, the most optimistic scenario is that (\ref{wmn_conjecture}) holds for generic (i.e. non-fine-tuned) gapped ground states, for a suitable definition of ``generic.'' In this sense, the R\'enyi modular commutator is on a similar footing as the original modular commutator.

The reader may wonder what we gain by having a R\'enyi generalization of the modular commutator. One advantage is that when $\alpha, \beta$ are positive integers, $\wab$ can be expressed in terms of an expectation value of a single operator acting on an appropriate replica system (see Eq.~\ref{replicaform} below). This replica representation provides a natural way to measure $\wab$ both experimentally and numerically, and also suggests a method for determining $\wab$ analytically, as in the case of the R\'enyi entanglement entropy.~\cite{islam2015measuring, Hastings2010renyi, Calabrese_2009}

This paper is organized as follows. In Sec.~\ref{sec:properties}, we discuss some of the general properties of $\wab$, as well as its connection to replica systems. In Sec.~\ref{sec:freefermions}, we present our computation of $\wab$ for gapped ground states of non-interacting fermion Hamiltonians. In Sec.~\ref{sec:string-nets}, we compute $\wab$ for string-net ground states. Finally, in Sec.~\ref{sec:discussion}, we summarize our results and discuss future directions.

\section{General properties of $\wab$}
\label{sec:properties}

\subsection{Consistency checks}
The R\'enyi modular commutator $\wab$ has several properties that are necessary for any consistent formula for $c_-$. First, $\wab$ is multiplicative under ``stacking'', i.e.~it is multiplicative under tensoring two decoupled 2D many-body systems in a bilayer geometry:
\begin{equation}
\wab(\rho\otimes\rho')=\wab(\rho)\wab(\rho').
\label{stackprop}
\end{equation}

Second, $\wab$ is odd under time reversal in the sense that
\begin{equation}
        \wab(\mathcal{T}\rho\mathcal{T}^{-1}) = \wab(\rho)^*,
\label{trprop}
\end{equation}
for any antiunitary time reversal transformation $\mathcal{T}$ of the form $\mathcal{T} = \mathcal{K} U$ where $U$ is a product of single-site unitaries and $\mathcal{K}$ denotes complex conjugation in a local basis.

Finally, $\wab = 1$ for any pure state that is supported on only three of the four regions, $A, B, C, D \equiv (ABC)^c$. That is:
\begin{align}
        \wab(|\psi_{ABC}\>) &= \wab(|\psi_{ABD}\>) = 1, \nonumber \\
        \wab(|\psi_{ACD}\> &=
        \wab(|\psi_{BCD}\>) = 1,
\label{tripartiteprop}
\end{align}
where $|\psi_{ABC}\>$ denotes a tripartite pure state that is supported only on $A, B, C$ and similarly for $|\psi_{ABD}\>$, $|\psi_{ACD}\>$, $|\psi_{BCD}\>$.

Properties (\ref{stackprop}) and (\ref{trprop}) are important consistency checks for any formula for $c_-$, since $c_-$ is additive under stacking and odd under time reversal. As for property (\ref{tripartiteprop}), this is a general requirement for any 2D topological many-body invariant since any such invariant must take the trivial value for a product state decorated with few-body tripartite states. 

We now discuss the derivation of these properties. The first two properties, (\ref{stackprop}) and (\ref{trprop}), follow straightforwardly from the definition of $\wab$. To derive (\ref{tripartiteprop}), suppose $\rho$ is a tripartite pure state, $\rho = |\psi_{ABC}\>\<\psi_{ABC}|$. Then $\rho_{AB}|\psi_{ABC}\> = \rho_C|\psi_{ABC}\>$ and $\rho_{BC}|\psi_{ABC}\> = \rho_A|\psi_{ABC}\>$, so 
\begin{align}
\<\rho_{AB}^\alpha\rho_{BC}^\beta\>=\<\rho_{C}^\alpha\rho_{A}^\beta\>.
\label{tripartpos}
\end{align}
Next, notice that the operator $\rho_{C}^\alpha\rho_{A}^\beta$ on the right hand side is positive semi-definite since it is a product of two commuting, positive semi-definite operators. Furthermore, one can check that $|\psi_{ABC}\>$ has no support on the null space of $\rho_{C}^\alpha\rho_{A}^\beta$. It follows that the right hand side of (\ref{tripartpos}) is  strictly positive so that $\wab = 1$. A similar argument applies for the other tripartite states, $|\psi_{ABD}\>$, $|\psi_{ACD}\>$ , and $|\psi_{BCD}\>$: in each case we can rewrite $\<\rho_{AB}^\alpha\rho_{BC}^\beta\>$ as an expectation value of the form $\<\rho_{R}^\alpha\rho_{R'}^\beta\>$ where $R, R'$ are disjoint regions.

\subsection{Replica representation}
We now show that when $\alpha=m,\ \beta=n$, where $m, n$ are positive integers, $\wmn$ can be written in terms of the expectation value of a permutation operator acting in an appropriate replica system. Specifically, consider a replica system built out of $m+n+1$ copies of our original state $\rho$. Then we can write
\begin{align}
\<\rho_{AB}^m\rho_{BC}^n\> = \Tr(\pi_A^{\sigma_a} \pi_B^{\sigma_b} \pi_C^{\sigma_c} \rho^ {\otimes (m+n+1)})  
\label{replicaform}
\end{align}
where $\pi^{\sigma_a}_A$, $\pi^{\sigma_b}_B$, $\pi^{\sigma_c}_C$ are replica permutation operators that permute the $m+n+1$ replicas of the lattice sites in regions $A, B, C$ according to the following (cyclic) permutations:
\begin{align}
  \sigma_a &= \bpm 1 & n+2 & n+3 & \cdots & m+n+1\epm  \nonumber \\
  \sigma_b &= \bpm 1 & 2 & \cdots & m + n + 1 \epm \nonumber \\
  \sigma_c &= \bpm 1 & 2 & \cdots & n + 1 \epm
\end{align}
Here, $\sigma_a$ and $\sigma_c$ are cyclic permutations of order $m+1$ and $n+1$ respectively, while $\sigma_b$ is a cyclic permutation of order $m + n + 1$. The precise definition of the replica permutation operators $\pi^\sigma_R$ is given by 
\begin{align}
    \pi^\sigma_R= \prod_{i \in R} \pi^\sigma_i
\end{align}
where the product runs over lattice sites $i \in R$ and where each $\pi^\sigma_i$ is a single-site permutation operator that permutes the $m+n+1$ replica sites associated with a single lattice site $i$. For bosonic systems, $\pi^\sigma_i$ is defined by
\begin{align}
\pi^\sigma | q_1 \cdots q_{m+n+1} \> = | q_{\sigma^{-1}}(1) \cdots q_{\sigma^{-1}}(m+n+1) \>
\end{align}
where $|q_1 \cdots q_{m+n+1}\>$ denotes a state in which the first replica lattice site is in some basis state $|q_1\>$, the second is in $|q_2\>$ and so on. (Here, we drop the $i$ subscript on $\pi^\sigma_i$ for notational simplicity, since the discussion that follows is completely focused on a single lattice site $i$). For fermionic systems, $\pi^\sigma$ can be defined by its action on the fermion annihilation operators $c_{1},...,c_{m+n+1}$, together with its action on the Fock space vacuum $|0\>$. Specifically, $\pi^\sigma$ leaves the replica vacuum invariant, i.e.~$\pi^\sigma |0\> = |0\>$, while it acts on the fermion operators as follows:
\begin{align}
    \pi^\sigma c_j (\pi^\sigma)^{-1} = \pm c_{\sigma(j)}.
\end{align}
The $\pm$ signs must be chosen carefully so that expectation values involving $\pi^\sigma$ on replica systems coincide with expectation values of reduced density operators. In particular, a consistent sign choice for the $\sigma$'s appearing in Eq.~\eqref{replicaform} is, with $\sigma=(i_1 i_2 \cdots i_k)$,
\begin{equation}
    \pi^\sigma c_{j} (\pi^\sigma)^{-1} =
    \begin{cases}
        -c_{\sigma(j)} & j\neq i_k\\
        c_{\sigma(j)} & j= i_k
    \end{cases}.
\end{equation}

We should mention that, while the above replica representation (\ref{replicaform}) is a powerful tool, we will not use it in this paper. Instead, our computations follow directly from the density matrix representation of $\wab$ (\ref{wmn_conjecture}), which applies to any real numbers $\alpha, \beta > 0$.

\section{Computation for non-interacting fermions}
\label{sec:freefermions}
In this section we prove Eq.~\eqref{wmn_conjecture} for gapped ground states of non-interacting fermion Hamiltonians. For clarity, we first present our derivation in the special case of charge conserving Hamiltonians (i.e.~complex fermion models). We then derive Eq.~\eqref{wmn_conjecture} in the more general case of Majorana fermion models.

\subsection{Quasidiagonality and the real-space Chern number}
\label{sec:qd_nuP}
We begin by explaining our basic setup and reviewing some useful concepts. Our main objects of interest are charge conserving fermionic Hamiltonians on an infinite two-dimensional lattice. These Hamiltonians take the form 
\begin{equation}
\label{complexH}
    \hat{H}=\sum_{j,k}h_{jk}c_j^\dagger c_k\qquad \{c_j,c_k^\dagger\}=\delta_{jk},
\end{equation}
where the single-particle Hamiltonian $h$ is an infinite-dimensional, Hermitian matrix, and we further assume finite-range hopping:
\begin{equation*}
    h_{jk}=0 \ \ \text{for} \ \ |j-k|\geq r_0,
\end{equation*}
for some finite $r_0$ \footnote{In fact, our results hold for a more general class of single-particle Hamiltonians, which only have exponential off-diagonal decay: $|h_{jk}|\leq be^{-\gamma |j-k|}$ for some $b,\gamma>0$.}. We also assume that the spectrum of $h$ is such that $\hat{H}$ has a finite energy gap (we need not assume any translational symmetry). 

The many-body ground state $|\Psi\>$ is obtained by filling all of the normal modes corresponding to the eigenvectors of $h$ with negative eigenvalues. This state is therefore fully determined by the ``spectral projector'' $P=\frac{1}{2}(\id-\operatorname{sgn}(h))$, an operator on the single-particle Hilbert space that projects onto these negative-energy modes. In particular, the two-point correlation function is given by the matrix elements of $P$: $\<\Psi|c_j^\dagger c_k|\Psi\>=P_{kj}$. Since $\hat{H}$ is gapped, the magnitude of $P_{jk}$ decays exponentially in the distance between $j$ and $k$ \cite{Hastings2006}:
\begin{equation}
\label{quasidiagonal_p}
    |P_{jk}|\leq b e^{-\gamma |j-k|},
\end{equation}
for some constants $b,\gamma>0$. We call operators with such decay properties ``quasidiagonal'' \footnote{This is a more restrictive definition of quasidiagonal than the one used in Ref.~\cite{Kitaev_2006}, which only assumes $|P_{jk}|\leq b |j-k|^{-\delta}$ for some $\delta > D$, the lattice dimension.}. In the next section, we will use the following properties of quasidiagonal operators:
\begin{enumerate}
    \item[(i)] If $X$ and $Y$ are quasidiagonal, their sum $X+Y$ and product $XY$ are as well.
    \item[(ii)] Let $X$ be quasidiagonal with spectrum $\sigma(X)$ and $f$ a holomorphic function on a bounded open region $D$ containing $\sigma(X)$. Then $f(X)$ is quasidiagonal. Moreover, $f(X)$ depends locally on $X$: for any quasidiagonal operator $V_R$ supported in region $R$ such that $\sigma(X+V_R)\subset D$, the matrix elements of $f(X+V_R)-f(X)$ decay exponentially outside of region $R$.
\end{enumerate}
See Appendix \ref{app:quasidiagonal} for derivations of these properties.

An important result that follows from the quasidiagonal property of $P$ is that it allows one to define a ``real-space Chern number" $\nu(P)$. For a tripartition of the lattice as shown in Fig.~\ref{fig:ABC_infinite}, the real-space Chern number is defined as~\cite{Kitaev_2006}
\begin{align}
\label{nu_kitaev}
    \nu(P)&=12\pi i \Tr\left(PAPBPCP-PCPBPAP\right)\nonumber\\
    &=4\pi i \Tr[PAP,PBP]\ ,
\end{align}
where $A,B,C$ denote spatial projections onto the corresponding regions. The quantity (\ref{nu_kitaev}) is a topological invariant, as moving a site from one region to another leaves \eqref{nu_kitaev} unchanged, and one can also show that it gives integer values. In the translationally invariant case, $\nu(P)$ coincides with the TKNN invariant~\cite{TKNN} defined as an integral over momentum space. For our purposes, we will need a generalization of Eq.~\eqref{nu_kitaev} that was introduced in Ref.~\cite{Fan_2023}. There, the authors showed that for positive integers $m$ and $n$,
\begin{align}
\label{nu_generalized}
    \Tr[(PAP)^m,(PBP)^n]&=\frac{\nu(P)}{2\pi i}\frac{m!\ n!}{(m+n)!}\nonumber\\
    &=\frac{i\nu(P)}{2\pi}\int_{0}^{1}t^m d(1-t)^n.
\end{align}
Note that this equation and the second equality in Eq.~\eqref{nu_kitaev} are only valid in the case where the size of $ABC$ is infinite, as otherwise the trace of the commutator is zero. 

\begin{figure}
    \centering
    \includegraphics[width=0.4\columnwidth]{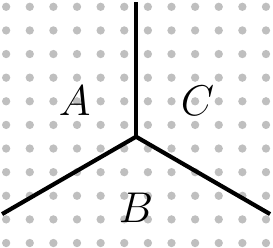}
    \caption{Tripartition of the lattice used in the definition of the real-space Chern number \eqref{nu_kitaev}.}
    \label{fig:ABC_infinite}
\end{figure}

\subsection{Computing \texorpdfstring{$\wab$}{Lg}}
\label{sec:computing_wmn}

We now compute $\wab$ for the ground state of the charge conserving Hamiltonian (\ref{complexH}). For simplicity, we will assume $\alpha$ and $\beta$ are positive integers in this computation. This case is simple because we can show that all error terms are exponentially small in the system size. Afterwards, we will comment on how the calculation can be generalized to all real $\alpha,\beta>0$.

\subsubsection*{Step 1: Single particle formula for $\wab$}

In the first step of our calculation, we find a single particle formula for $\wab$. More specifically, we express $\wab$ in terms of the spectral projector $P =\frac{1}{2}(\id-\operatorname{sgn}(h))$. 

To find the desired formula, consider the expectation value $\<\rho_{AB}^\alpha\rho_{BC}^\beta\>$. Using standard results, one can show that $\<\rho_{AB}^\alpha\rho_{BC}^\beta\>$ is given by the following determinant:
\begin{equation}
\label{mdef}
\begin{gathered}
    \<\rho_{AB}^\alpha\rho_{BC}^\beta\>=\det M,\\
    M=(\id-P)+P(P_{AB}+P_{CD})^\alpha(P_{BC}+P_{AD})^\beta P.
\end{gathered}
\end{equation}
Here $P_R$ denotes the restriction $P_R = RPR$ where $R$ is the spatial projection onto region $R$. For a detailed derivation of (\ref{mdef}), see Appendix \ref{app:m_derivation}.

We obtain a formula for $\wab$ by taking the polar decomposition $M=U|M|$, where $|M|=\sqrt{M^\dagger M}$. This allows us to write Eq.~\eqref{wmn_definition} compactly as a single determinant of the unitary part:
\begin{equation}
\label{wmn_compact}
    \wab= \frac{\det M}{\det|M|}=\det U.
\end{equation}
We show in Appendix \ref{app:m_invertibility} that $M$ is always invertible, and therefore the determinant in Eq.~\eqref{mdef} is nonzero, for \emph{any} spectral projector $P$. It follows that $\det U$ is always a well-defined $U(1)$ phase in the case of non-interacting fermion systems\footnote{Here we are using the notion of the ``Fredholm determinant'', which is a generalization of the regular determinant for finite matrices to the infinite dimensional case. It is defined for any operator of the form $\id+T$ where $T$ has a well-defined trace. The matrix $M$ is indeed of this form, as $M-\id$ decays exponentially to zero away from the $ABC$ region.}. Our task is now to evaluate $\det U$.

\subsubsection*{Step 2: Truncating to the triple points}

Next, we claim that $\det U$ only receives contributions from the neighborhood of the four triple points of the $ABCD$ partition. Explicitly, we define four large disjoint regions of linear size $l_\Pi$ that surround each of the triple points and label them as $\Pi_i$ where $i$ goes from 1 to 4, as in Fig.~\ref{fig:ABCD}. We claim that
\begin{multline}
\label{u_truncation}
    \det U=\left(\rdet{\Pi_1} U\right)\left(\rdet{\Pi_2} U\right)\left(\rdet{\Pi_3} U\right)\left(\rdet{\Pi_4} U\right)
\end{multline}
where $\det_{\Pi_i}$ indicates that the determinant is restricted to the indices corresponding to sites in $\Pi_i$ and where equality holds up to corrections that are exponentially small in $l_\Pi$ \footnote{More precisely, Eq.~\eqref{u_truncation} holds up to errors of order $O(e^{-\gamma' l_\Pi})$, where $\gamma'$ is the decay coefficient of the off-diagonal elements of $U$. Using techniques from Appendix \ref{app:quasidiagonal}, the value of $\gamma'$ can be bounded in terms of the original $\gamma$ appearing in Eq.~\eqref{quasidiagonal_p} and the powers $\alpha, \beta$. In particular, one finds that $\gamma'$ is equal to the decay coefficient of $|M|^{-1}$, which can be lower bounded in terms of $\norm{M^{-1}}^{-1}$. Explicitly, we have
\begin{equation*}
    \gamma' \geq \frac{\gamma}{\norm{M^{-1}}}\geq \frac{\gamma}{2^{2(\alpha+\beta+1)}},
\end{equation*}
where we utilize the bound on $\norm{M^{-1}}$ derived in Appendix \ref{app:m_invertibility}.}.

\begin{figure}
    \centering
    \includegraphics[width=0.5\linewidth]{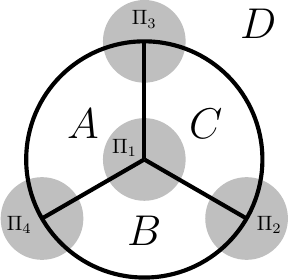}
    \caption{Depiction of the spatial regions $\Pi_i$ used in Eq.~\eqref{u_truncation}, which surround each of the four triple points in the $ABCD$ partition. Outside of these regions, $U$ is approximately equal to the identity.}
    \label{fig:ABCD}
\end{figure}

Equation \eqref{u_truncation} is a consequence of the following two properties of $M$ and $U$:
\begin{enumerate}
    \item\label{prop1} $U$ is quasidiagonal and depends locally on $M$ in the same sense as in property (ii) of the previous section.
    \item\label{prop2} $M$ is ``locally positive definite'' away from the triple points: let $R$ be a subregion that intersects at most two of the $ABCD$ regions, and let $d_{R}$ be the distance between it and the other two regions.
    Then there exists a quasidiagonal positive definite matrix $M_{pd}$ such that 
    \begin{equation}
    \label{locallyPD}
        RMR=RM_{pd}R+O(e^{-\gamma d_{R}}),
    \end{equation}
    where $\gamma$ is the decay coefficient of the spectral projector given in Eq.~\eqref{quasidiagonal_p} and $R$ denotes the projection onto region $R$.
\end{enumerate}
To see why these properties imply Eq.~\eqref{u_truncation}, note that
property \ref{prop2} along with local dependence from property \ref{prop1} implies that, up to errors exponentially small in $d_R$,
\begin{equation}
    RUR=R.
\end{equation}
This follows because the unitary part of the polar decomposition of $M_{pd}$ is the identity. We therefore see that $U$ rapidly approaches the identity outside of the $\Pi_i$ regions, so up to errors exponentially small in the size of the $\Pi_i$,
\begin{equation}
    \det U=\rdet{(\Pi_1\cup\Pi_2\cup\Pi_3\cup\Pi_4)}U. 
\end{equation}
Quasidiagonality of $U$ then further implies that $U$ block-diagonalizes on the $\Pi_i$ subspaces, and Eq.~\eqref{u_truncation} follows. 

We now prove properties \ref{prop1} and \ref{prop2}. For property \ref{prop1}, note that $(P_{AB}+P_{CD})^\alpha$ and $(P_{BC}+P_{AD})^\beta$ are quasidiagonal by property (i) of the previous section, so $M$ and $M^\dagger$ must be as well. (Remark: this step is where our assumption that $\alpha$ and $\beta$ are integers comes into play: if $\alpha, 
\beta > 0$ are not integers, then we cannot use property (i) here, and our argument that $M$ is quasidiagonal breaks down. Instead, we can only show that the matrix elements $M_{jk}$ decay as a \emph{power-law} in $|j-k|$ -- see comment at the end of this section.) Further, we show in Appendix \ref{app:m_invertibility} that $M$ is invertible with spectrum bounded away from zero by a constant that depends only on $\alpha$ and $\beta$. The spectrum of $M^\dagger M$ is therefore similarly bounded away from zero and positive, so property (ii) with $f(x)=x^{-1/2}$ implies that $|M|^{-1}$ is quasidiagonal and depends locally on $M$. Property \ref{prop1} then follows from (i) and the fact that $U=M|M|^{-1}$.

To derive property \ref{prop2}, let $k$ be a site located inside the $AB$ region. Then for any other site $j$, Eq.~\eqref{quasidiagonal_p} implies
\begin{equation*}
    (P_{AB}+P_{CD})_{jk}=P_{jk}+O\left(e^{-\gamma \operatorname{dist}(k,CD)}\right),
\end{equation*}
and
\begin{align*}
    (P_{BC}+P_{AD})_{jk}&=(P_{B}+P_{A})_{jk}+O\left(e^{-\gamma \operatorname{dist}(k,CD)}\right)\nonumber\\
    &=(P_{A^c}+P_{A})_{jk}+O\left(e^{-\gamma \operatorname{dist}(k,CD)}\right).
\end{align*}
We therefore see that for $k$ deep enough in the interior of $AB$, the matrix element $M_{jk}$ satisfies:
\begin{multline}
    M_{jk}=\left[(\id-P)+P(P_{A^c}+P_{A})^\beta P\right]_{jk}\\
    +O\left(e^{-\gamma\operatorname{dist}(k,CD)}\right).
\end{multline}
The operator appearing on the right-hand side is positive definite, and one can check that a similar form holds for $M_{jk}$ in all other areas where $k$ is chosen to be far from two of the $ABCD$ regions. 

\subsubsection*{Step 3: Reducing to a tripartition of the infinite plane}
We now claim that each of the four truncated determinants in Eq.~\eqref{u_truncation} can be written as the determinant of a unitary operator $\uabc$ given by
\begin{equation}
\label{udef_ABC}
\begin{gathered}
    \uabc=\mabc|\mabc|^{-1},\\
    \mabc=(\id-P)+P(P_{\mathcal{A}\mathcal{B}}+P_{\mathcal{C}})^\alpha(P_{\mathcal{B}\mathcal{C}}+P_{\mathcal{A}})^\beta P,
\end{gathered}
\end{equation}
where $\abc$ is a tripartition of the plane that is topologically equivalent to Fig~\ref{fig:ABC_infinite}. Moreover, $\det\uabc$ gives the same value for all topologically equivalent choices of $\abc$, so that 
\begin{equation}
\label{uabc_claim}
    \det U=\left(\det\uabc\right)^4.
\end{equation}
(Here, and in what follows, we neglect exponentially small error terms like the one in Eq.~\eqref{u_truncation}.)

\begin{figure}
     \subfloat[\label{fig:det1}]{\includegraphics[width=4cm]{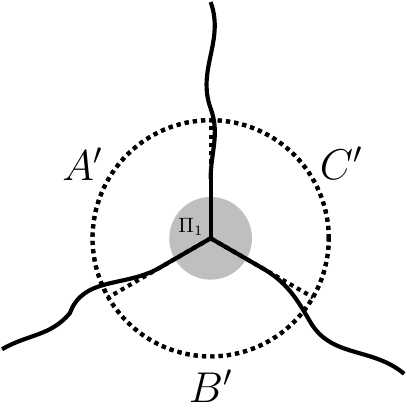}}
     \hfill
     \subfloat[\label{fig:det2}]{\includegraphics[width=4cm]{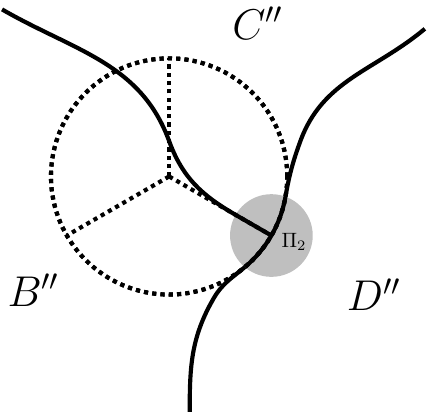}}
     \caption{Tripartitions of the plane used to compute (a) $\rdet{\Pi_1}U$ and (b) $\rdet{\Pi_2}U$. The dotted lines show the boundaries of the original $ABC$ region.}
\end{figure}
We now derive these claims. We start with $\det_{\Pi_1} U$. This determinant is invariant under arbitrary deformations to the $ABCD$ partition away from $\Pi_1$. In particular, we can consider new regions $A',B',C'$ like those shown in Fig.~\ref{fig:det1}, which form a tripartition of the plane and satisfy
\begin{equation*}
    A'|_{\Pi_1}=A|_{\Pi_1},\quad B'|_{\Pi_1}=B|_{\Pi_1},\quad C'|_{\Pi_1}=C|_{\Pi_1}.
\end{equation*}
We then define a new operator $M'$ by taking $\{A,B,C,D\}\rightarrow \{A',B',C',\emptyset\}$ in Eq.~\eqref{mdef}. Likewise, we define a new unitary $U'$ by $U' = M' |M'|^{-1}$. Quasidiagonality implies that $\Pi_1M\Pi_1=\Pi_1M'\Pi_1$ provided that the boundaries of $A'B'C'$ deviate from $ABC$ slowly enough outside of $\Pi_1$ \footnote{Concretely, we can choose $A',B',C'$ to agree with $ABC$ on a larger set $\Pi_1'$ that is defined by thickening $\Pi_1$ by some length $l'$. We then have $\Pi_1M\Pi_1=\Pi_1M'\Pi_1+O(e^{-\gamma l'})$ from quasidiagonality}. From property \ref{prop1} we therefore have
\begin{equation}
    \rdet{\Pi_1}U = \rdet{\Pi_1}U'.
\end{equation}
At the same time, if we apply the truncation result (\ref{u_truncation}) to $U'$, then since there is only one triple point in this configuration (Fig.~\ref{fig:det1}), we deduce that
\begin{equation}
    \det U'=\rdet{\Pi_1}U'.
\end{equation}
Combining these two results, we find that $\rdet{\Pi_1} U = \det U'$, where $U'$ is of the form given in Eq.~\eqref{udef_ABC}. This explains the first factor of $\det\uabc$ in (\ref{uabc_claim}).

Now consider the second determinant in \eqref{u_truncation}. Following the same arguments, we find that $\rdet{\Pi_2} U=\det U''$, where $U''$ is the unitary part of the matrix
\begin{equation*}
    M''=(\id-P)+P(P_{B''}+P_{C''D''})^\alpha(P_{B''C''}+P_{D''})^\beta P,
\end{equation*}
and $D''C''B''$ is a tripartition like that shown in Fig.~\ref{fig:det2}. Relabeling $\{D'',C'',B''\}$ as $\{\mathcal{A},\mathcal{B},\mathcal{C}\}$, we find that $U''$ is also of the form given in Eq.~\eqref{udef_ABC}. 

The remaining two determinants in \eqref{u_truncation} follow a similar analysis. Therefore, to finish the derivation of Eq.~\eqref{uabc_claim}, all that remains is to show that the four determinants coming from the different triple points are all the same. Equivalently, we need to show that $\det\uabc$ is a topological invariant. We show this by directly computing $\det\uabc$. We will see that $\det\uabc$ is independent of the choice of $\mathcal{A}, \mathcal{B}, \mathcal{C}$ and is proportional to the real-space Chern number.

\subsubsection*{Step 4: Evaluating $\det\uabc$}
At this point, all that remains is to evaluate $\det\uabc$ since $\wab = (\det \uabc)^4$ according to \eqref{wmn_compact} and \eqref{uabc_claim}. We now sketch this calculation -- see Appendix \ref{app:determinant_derviation} for more details.

Our strategy is to express $\mabc$ as a product of three exponentials,
\begin{equation}
\label{mabc_decomposition}
    \mabc=e^{O_1}e^{O_2}e^{O_3},
\end{equation}
where the $O_i$ operators have two nice properties: (1) they are \emph{Hermitian} and (2) the commutators $[O_i,O_j]$ have well-defined traces. These properties allow us to express the unitary part $\uabc$ in terms of commutators of the $O_i$ via the Baker-Campbell-Hausdorff expansion, from which we can compute the determinant using Eq.~\eqref{nu_generalized}.

We begin by expanding out the second term in $\mabc$ as
\begin{align}
P(P_{\mathcal{A}\mathcal{B}} +P_{\mathcal{C}})^\alpha&(P_{\mathcal{B}\mathcal{C}}+P_{\mathcal{A}})^\beta P = \nonumber \\
&P(P_{\mathcal{A}\mathcal{B}}^\alpha P_{\mathcal{B}\mathcal{C}}^\beta + 
P_{\mathcal{C}}^\alpha P_{\mathcal{B}\mathcal{C}}^\beta + P_{\mathcal{A}\mathcal{B}}^\alpha P_{\mathcal{A}}^\beta )P
\label{m2ndterm}
\end{align}
Next we use the following operator identities which hold for any exponent $x > 0$:
\begin{equation*}
P (RPR)^x = (PRP)^x R, \quad (RPR)^x P = R (PRP)^x.
\end{equation*}
Applying these identities to (\ref{m2ndterm}), we rewrite these terms as
\begin{align}
&(P-\rc)^\alpha(P-\ra-\rc)(P-\ra)^\beta \nonumber\\
&+\rc^{\alpha+1}(P-\ra)^\beta + (P-\rc)^\alpha \ra^{\beta+1}
\end{align}
where we are using the abbreviation $\rp{R}=PRP$ where $R$ is a spatial projector. Then, including the first term in $\mabc$, (namely $\id - P$), we obtain
\begin{align}
    \mabc&=(\id-\rc)^\alpha(\id-\ra-\rc)(\id-\ra)^\beta \nonumber\\
    &\quad+\rc^{\alpha+1}(\id-\ra)^\beta+(\id-\rc)^\alpha\ra^{\beta+1}.
\end{align}
We then obtain the factorization given in Eq.~\eqref{mabc_decomposition} by letting
\begin{align}
\label{oi_def}
    O_1&=\ln\left(\id-\rc\right)^\alpha\nonumber\\
    O_2&=\ln\left(\id-\ra-\rc+\frac{\rc^{\alpha+1}}{(\id-\rc)^{\alpha}}+\frac{\ra^{\beta+1}}{(\id-\ra)^{\beta}}\right)\\
    O_3&=\ln\left(\id-\ra\right)^\beta\nonumber.
\end{align}

We note that the operators $O_1, O_2, O_3$ are generally unbounded, since $\id-\ra$ and $\id-\rc$ can have arbitrarily small or even vanishing eigenvalues. In the following derivation we will ignore this issue, effectively interpreting the $O_i$ as formal power series expansions. For a more careful treatment, including a proper regularization of the $O_i$, see Appendix \ref{app:determinant_derviation}.

We now use the Baker-Campbell-Hausdorff expansion to write
\begin{equation}    \uabc=e^{O_1}e^{O_2}e^{O_3}\left|e^{O_1}e^{O_2}e^{O_3}\right|^{-1}\equiv e^{iK},
\end{equation}
for a Hermitian $K$, which can be expressed as a formal sum of commutators between the $O_i$. One can show that only the lowest order commutators in $K$ have a nonvanishing trace (at least if we regularize the $O_i$'s as in Appendix \ref{app:determinant_derviation}) so $\det\uabc$ reduces to the following expression:
\begin{align}
\label{uabc_traces}
    \det\uabc&=\exp\Tr(iK)\nonumber\\
    &=\exp\frac{\Tr[O_1,O_2]+\Tr[O_2,O_3]+\Tr[O_1,O_3]}{2}.
\end{align}

To compute these traces, we note that Eq.~\eqref{nu_generalized} implies, for functions $f$ and $g$ that are analytic at $0$, 
\begin{equation}
\label{integralform}
    \Tr[f(\rc), g(\ra)]=\frac{i\nu(P)}{2\pi}\int_{0}^{1} [f(t)-f(0)] \partial_{t} g(1-t) dt.
\end{equation}
We can immediately use this to calculate the third trace in Eq.~\eqref{uabc_traces}:
\begin{align}
    \Tr[O_1, O_3]&=\Tr\left[\ln \left(\id-\rc\right)^\alpha, \ln \left(\id-\ra\right)^\beta\right] \nonumber\\
    & =\alpha\beta\frac{i\nu(P)}{2\pi}\int_{0}^{1} \frac{\ln (1-t)}{t} d t \nonumber\\
    &=-\alpha\beta\frac{\pi i\nu(P)}{12}.
\end{align}
The other two traces are more involved, and are computed in Appendix~\ref{app:determinant_derviation}. The main idea is that $\Tr[O_1,O_2]$ and $\Tr[O_2,O_3]$ can each be broken up into components for which we can use Eq.~\eqref{integralform}. The result is 
\begin{equation}
    \Tr[O_1, O_2]
    =-\frac{\pi i\nu(P)}{24}\alpha\left(\frac{1}{\alpha+1}-\frac{1}{\alpha+\beta+1}-\beta\right)
\end{equation}
and 
\begin{equation}
    \Tr[O_2, O_3]
    =-\frac{\pi i\nu(P)}{24}\beta\left(\frac{1}{\beta+1}-\frac{1}{\alpha+\beta+1}-\alpha\right).
\end{equation}
Inserting these values into Eq.~\eqref{uabc_traces}, we obtain
\begin{equation}
    \det \uabc=\exp\left\{-\frac{\pi i}{48}q(\alpha,\beta)\nu(P)\right\}.
\end{equation}
Finally, applying equations \eqref{wmn_compact} and \eqref{uabc_claim}, we conclude that $\wab$ takes the value given in Eq.~\eqref{wmn_conjecture} since $c_-=\nu(P)$ for complex fermions.

\subsubsection*{Extending to real \texorpdfstring{$\alpha,\beta>0$}{Lg}}

We now sketch how to generalize the above derivation to the case where $\alpha, \beta$ are arbitrary positive real numbers. As we mentioned earlier, the main complication is that in the general case, $M$ need not be quasidiagonal in the sense of (\ref{quasidiagonal_p}) -- that is, the matrix elements $M_{jk}$ need not decay exponentially with $|j-k|$. Instead, one can prove a \emph{power-law} bound: one can show that $|M_{jk}| \leq c (1+|j-k|)^{-\delta}$ for some constant $c$ and some $\delta > 1$. In addition, one can show that $|M_{jk} -\id_{jk}|$ decays exponentially with the distance from $j, k$ to the boundaries of regions $A, B, C$. These two properties can be established with similar tools to the ones used in this paper, e.g.~the holomorphic functional calculus approach of Appendix~\ref{app:quasidiagonal}.

We believe that Steps 2-4 can still be justified using these weaker bounds on $M_{jk}$. The point is that we can weaken our definition of quasidiagonality to include any operator $A$ such that $|A_{jk}| \leq c (1+|j-k|)^{-\delta}$ for some $\delta > D$, the dimension of the lattice. The above bounds on $M_{jk}$ mean that $M$ is effectively quasidiagonal in this weaker sense. Then, using properties (i) and (ii), which still hold for this weaker notion of quasidiagonality \cite{jaffard_1990}, one can argue that the unitary $U =  M |M|^{-1}$ is also quasidiagonal in this weaker sense. Following the same logic, the truncation in Eq.~\eqref{u_truncation} and the reduction to the tripartition in Eq.~\eqref{uabc_claim} should still hold, but with errors that decay with a power-law in $l_\Pi$ rather than exponentially in $l_\Pi$. In this way, one can argue that the calculation is still valid for general $\alpha, \beta > 0$, but with \emph{power-law} bounds on error terms.

\subsection{Majorana fermions}
\label{sec:majorana}
We now derive Eq.~(\ref{wmn_conjecture}) for gapped Majorana fermion models. Specifically, we consider gapped quadratic fermion Hamiltonians on the infinite 2D lattice that take the form 
\begin{equation}
\label{hdef_majorana}
    \hat{H}=\frac{i}{4}\sum_{j,k}h_{jk}\chi_j \chi_k\qquad \{\chi_j,\chi_k\}=2\delta_{jk},
\end{equation}
where $h$ is a real skew-symmetric matrix with finite-range hopping. 

As in the complex fermion case, the first step of the calculation is to find an expression for $\<\rho_{AB}^\alpha\rho_{BC}^\beta\>$ in terms of the spectral projector on the single-particle Hilbert space. We outline how this is done in Appendix \ref{app:majorana}: the main difference with the complex case is that expectation values of Gaussian operators are more naturally expressed in terms of Pfaffians, rather than determinants. For this reason, it is more straightforward to instead compute \emph{squares} of expectation values, which can then be expressed in the more familiar determinant formalism. In particular, we show in Appendix \ref{app:majorana} that for the Majorana fermion model (\ref{hdef_majorana}),
\begin{equation}
\label{mdef_majorana}
\begin{gathered}
    \<\rho_{AB}^\alpha\rho_{BC}^\beta\>^2=\det M,\\
    M=(\id-P)+P(P_{AB}+P_{CD})^\alpha(P_{BC}+P_{AD})^\beta P.
\end{gathered}
\end{equation}
where $P$ is the spectral projector $P=\frac{1}{2}(\id-\operatorname{sgn}(h))$.

Next, notice that the above expression for $M$ (\ref{mdef_majorana}) has exactly the same form as that appearing in Eq.~\eqref{mdef}, and the spectral projector satisfies the same quasidiagonal properties in the Majorana basis as it does in the complex basis. We can therefore follow the same logic as in the complex case to derive
\begin{equation}
    \wab^2=\exp\left\{-\frac{\pi i}{12}q(\alpha,\beta)\nu(P)\right\}
\end{equation}
for Majorana fermions. Taking the square root introduces a $e^{\pi i}$ phase ambiguity, which is fixed by requiring that $\omega_{\alpha,0}=\omega_{0,\beta}=1$, leaving us with 
\begin{equation}
    \wab=\exp\left\{-\frac{\pi i}{24}q(\alpha,\beta)\nu(P)\right\}.
\end{equation}
The extra factor of $\frac{1}{2}$ in the exponent reflects the fact that each Majorana mode constitutes only half of a complex fermion mode. Correspondingly, we have $c_-=\frac{\nu(P)}{2}$ when $P$ is defined the Majorana basis, and we conclude that the value of $\wab$ agrees with \eqref{wmn_conjecture} in the Majorana case as well.

\section{Computation for string-net models}
\label{sec:string-nets}

In this section we show that $\wab = 1$ for string-net ground states. This result is consistent with Eq.~(\ref{wmn_conjecture}) since all string-net ground states have $c_- = 0$, as we review below.

\subsection{Review of string-net models}

We begin with a brief review of string-net models.~\cite{Levinwen2005, hong2009symmetrization, Hahn2020, LinLevinBurnell2021} String-net models are exactly solvable 2D lattice models that realize a large class of topological phases with nontrivial anyon excitations. Specifically, these models are believed to realize every topological phase whose boundary with the vacuum can be gapped.~\cite{KitaevKong} 

The input data necessary to build a string-net model is a set of string types, $a \in \mathcal{C}$, together with a corresponding set of quantum dimensions $\{d_a: a \in \mathcal{C}\}$, branching rules, and $F$-symbols obeying certain consistency conditions. The output of the string-net construction is a lattice spin model with nontrivial anyon excitations. 

\begin{figure}
    \centering    
    \includegraphics[width=0.6\columnwidth]{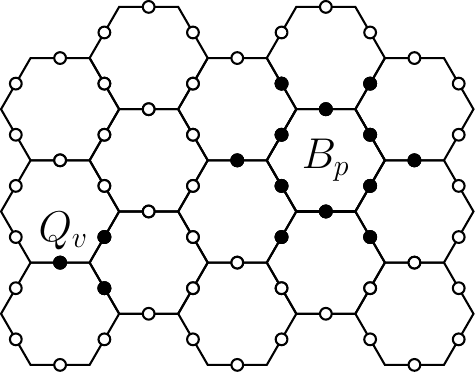}
    \caption{The two terms in the string-net Hamiltonian (\ref{snet-Ham}). The regions of support of the $Q_v$ and $B_p$ operators are indicated by filled circles.}
    \label{fig:snetlattice}
\end{figure}

Concretely, string-net models are built out of finite-dimensional spins (i.e.~qudits) living on the links of the honeycomb lattice (Fig.~\ref{fig:snetlattice}). Each spin can be in a finite number of basis states $|a\>$ -- one for each string type $a \in \mathcal{C}$. The Hamiltonian takes the form
\begin{equation}
H = - \sum_v Q_v - \sum_p B_p
\label{snet-Ham}
\end{equation}
where $v$ and $p$ label the vertices and plaquettes of the honeycomb lattice and where $Q_v$ and $B_p$ are commuting Hermitian projection operators:
\begin{align}
    [Q_v, Q_{v'}] = [B_p, B_{p'}] = [Q_v, B_p] = 0
\end{align}
More specifically, $Q_v$ is a projection operator that acts on the three links adjacent to vertex $v$ and is diagonal in the string basis $\{|a\>\}$. The $B_p$ operator is a projector that acts on the six links along the boundary of plaquette $p$ as well as the six links that touch $p$, which we refer to as the ``legs'' of $p$ (Fig.~\ref{fig:snetlattice}). This operator acts diagonally in the string basis $\{|a\>\}$ on the six legs of $p$ (though it does not act diagonally on the six links along the boundary of $p$). The specific form of $Q_v, B_p$ can be found in Refs.~\cite{Levinwen2005, hong2009symmetrization, Hahn2020, LinLevinBurnell2021}, but is not necessary for our calculation.

An important property of the Hamiltonian (\ref{snet-Ham}) is that it has a unique gapped ground state $|\Phi\>$ in an infinite plane geometry. This state $|\Phi\>$ is the unique state that obeys 
\begin{align}
Q_v |\Phi\> = B_p |\Phi\> = |\Phi\>
\end{align}
for every vertex $v$ and plaquette $p$. 

Another key property of the Hamiltonian (\ref{snet-Ham}) is that its ground state $|\Phi\>$ has vanishing chiral central charge. Indeed, ground states of local commuting Hamiltonians always have $c_- = 0$, by the analysis given in Appendix D.1 of Ref.~\cite{Kitaev_2006}.\footnote{For a commuting Hamiltonian, the energy current $f_{jk} = 0$ in Eq.~(154) of Ref.~\cite{Kitaev_2006}, and hence one can choose $h_{jkl} = 0$ in Eq.~(159), which leads to a vanishing $c_-$ in Eq.~(160).} 

\subsection{Doubling the link degrees of freedom}

For technical reasons, it is simpler to compute $\wab$ for a slightly modified string-net ground state $|\tilde{\Phi}\>$ that has \emph{two} spins on each link (Fig.~\ref{fig:doubled_lattice}). This state $|\tilde{\Phi}\>$ is obtained from $|\Phi\>$ by making the replacement $|a\> \rightarrow |aa\>$ for each basis state $|a\>$ on each link $l$ of the honeycomb lattice. 

\begin{figure}
    \centering    
    \includegraphics[width=0.95\columnwidth]{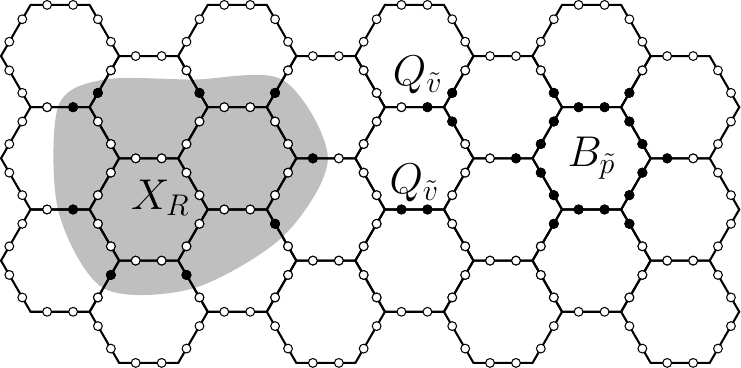}
    \caption{String-net model on the decorated honeycomb lattice. The $Q_{\tilde{v}}$ operator acts on either two or three spins around each vertex $\tilde{v}$. The $B_{\tilde{p}}$ operator acts on the 18 spins adjacent to the plaquette $\tilde{p}$. We focus on regions $R$ (shaded region) with a boundary $\partial R$ that separates the two spins on each boundary link; the operator $X_R$ acts on the spins just inside $R$, indicated by filled circles.}
    \label{fig:doubled_lattice}
\end{figure}

Conveniently, this modified state $|\tilde{\Phi}\>$ is also the ground state of a string-net Hamiltonian of the form (\ref{snet-Ham}), but on a \emph{decorated} honeycomb lattice. This decorated lattice is obtained by dividing each link of the original lattice into two links in the decorated lattice, placing additional (bivalent) vertices in between. We denote the vertices, links, and plaquettes of the decorated honeycomb lattice by $\tilde{v}, \tilde{l}, \tilde{p}$ respectively. (Here the plaquettes $\tilde{p}$ coincide with the plaquettes $p$ in the original honeycomb lattice). 

In this notation, the modified state $|\tilde{\Phi}\>$ is the ground state of a Hamiltonian of the form 
\begin{align}
H = - \sum_{\tilde{v}} Q_{\tilde{v}} - \sum_{\tilde{p}} B_{\tilde{p}}, 
\end{align}
where $Q_{\tilde{v}}$ and $B_{\tilde{p}}$ have the same structure as before: i.e.~$Q_{\tilde{v}}$ is supported on the two or three links $\tilde{l}$ adjacent to each vertex $\tilde{v}$, while $B_{\tilde{p}}$ is supported on the $12$ links $\tilde{l}$ along the boundary of $\tilde{p}$ as well as the $6$ ``legs'' of $\tilde{p}$ (Fig.~\ref{fig:doubled_lattice}).

The main advantage of working with the string-net state $|\tilde{\Phi}\>$ on the decorated honeycomb lattice is that it allows for more symmetrical partitions between different subregions $A, B, C$ of the lattice. In particular, we can choose each region $R$ so that the boundary $\partial R$ separates the two spins located on each boundary link (Fig.~\ref{fig:doubled_lattice}). In our calculation, we will assume that all regions $R = A, B, C$ have symmetrical boundaries of this type.

\subsection{Structure of $\rho_R$ for string-net ground states}
The key result that we will need to compute $\wab$ is the structure of the reduced density operator $\rho_R$ for simply-connected regions $R$. This structure was worked out by Ref.~\cite{levinwen2006} for the modified string-net ground state $|\tilde{\Phi}\>$. It was shown there that $\rho_R$ takes the form
\begin{align}
    \rho_R = X_R P_R 
    \label{rhoRidentity}
\end{align}
where $P_R$ is a projection operator and $X_R$ is an operator acting on the links $\tilde{l}$ along the boundary of $R$. More specifically, $P_R$ is defined as the product of all $Q_{\tilde{v}}$ and $B_{\tilde{p}}$ projectors that are supported entirely in $R$:
\footnote{In Ref.~\cite{levinwen2006}, $P_R$ was written in a different form -- namely as a projection onto the subspace of states $\{|\Phi^{\text{in}}_{q,s}\>\}$ obeying certain ``local rules'' in region $R$. In (\ref{rhoRidentity}) we use the fact that imposing the local rules is equivalent to imposing $Q_{\tilde{v}} = B_{\tilde{p}} = 1$ within region $R$.} 
\begin{align}
    P_R = \prod_{\text{supp}(Q_{\tilde{v}}) \subseteq R} Q_{\tilde{v}} \prod_{\text{supp}(B_{\tilde{p}}) \subseteq R} B_{\tilde{p}}
\label{Prdef}
\end{align}
The operator $X_R$ is defined as a product of single-link operators $X_{\tilde{l}}$ acting on the links $\tilde{l}$ along the boundary of $R$ (Fig.~\ref{fig:doubled_lattice}):
\begin{align}
    X_R = N_R \prod_{\tilde{l} \in \partial R} X_{\tilde{l}}, 
\label{Xrdef}
\end{align}
Here $N_R$ is the (positive) normalization constant $N_R = \frac{1}{D^{|\partial p|-1}}$ with $D = \sum_a d_a^2$, while the operator $X_{\tilde{l}}$ is the single-link operator acting on link $\tilde{l}$ that is diagonal in the string basis and has matrix elements 
\begin{align}
\<a' |X_{\tilde{l}} |a\> = d_a \delta_{aa'}
\end{align}

In our computation below, we will need the following properties of $X_R, P_R$:
\begin{enumerate}
\item{$X_R$ is positive definite.}\label{xpprop1}
\item{$[X_R, X_{R'}] = 0$ for any two regions $R, R'$.}\label{xpprop2}
\item{$P_R P_{R'} = P_{R'}$ if $R \subseteq R'$}\label{xpprop3}
\item{$[P_R, X_{R'}] = 0$ if $R \subseteq R'$.}\label{xpprop4}
\end{enumerate}
Properties \ref{xpprop1} and \ref{xpprop2} follow immediately from the definition of $X_R$ (\ref{Xrdef}) together with the fact that the quantum dimensions $d_a$ are positive for all $a$. Property \ref{xpprop3} follows from the fact that $Q_{\tilde{v}}$ and $B_{\tilde{p}}$ are commuting projectors. As for property \ref{xpprop4}, recall that any $B_{\tilde{p}}$ that appears in $P_R$ (\ref{Prdef}) is supported entirely in $R$. It follows that the plaquette boundary $\partial \tilde{p}$ cannot contain any $\tilde{l} \in \text{supp}(X_{R'})$ for any $R' \supseteq R$. Hence the only way that the $B_{\tilde{p}}$ operators in $P_R$ can act nontrivially on the links $\tilde{l} \in \text{supp}(X_{R'})$ is if $\tilde{l}$ is one of the ``legs'' of $\tilde{p}$. Then since $B_{\tilde{p}}$ acts diagonally in the string basis on the legs of $\tilde{p}$, we conclude that $[B_{\tilde{p}}, X_{R'}] = 0$ for any $B_{\tilde{p}}$ that appears in $P_R$. Furthermore, it is clear that $[Q_{\tilde{v}}, X_{R'}] = 0$ for all $\tilde{v}$ since $Q_{\tilde{v}}$ is diagonal in the string basis. Putting this all together gives $[P_R, X_{R'}] = 0$, as we wished to show. 

\subsection{Computing $\wab$}
We now use the above properties \ref{xpprop1}-\ref{xpprop4} to show that $\wab = 1$ for string-net models. First, we note that $\rho_R^\alpha$ can be rewritten as
\begin{align}
    \rho_R^\alpha &= \left(X_{R} P_{R} \right)^\alpha \nonumber \\
    &= X_R^\alpha P_R^\alpha
\end{align}
where the second line follows from the fact that $[P_R, X_R] = 0$ (property \ref{xpprop4} above). Next we note that
\begin{align}
\<\rho_{AB}^\alpha \rho_{BC}^\beta\> &= \Tr(\rho_{AB}^\alpha \rho_{BC}^\beta \rho_{ABC}) \nonumber \\
&= \Tr\left[\left(X_{AB}^\alpha P_{AB}^\alpha  \right) \left(X_{BC}^\beta P_{BC}^\beta \right) \left(X_{ABC} P_{ABC} \right)\right] \nonumber \\
&= \Tr( X_{AB}^\alpha X_{BC}^\beta X_{ABC} P_{BC}^\beta P_{ABC} P_{AB}^\alpha ) \nonumber \\
&= \Tr( X_{AB}^\alpha X_{BC}^\beta X_{ABC} P_{ABC})
\label{tracesnet}
\end{align}
where the third equality follows from property \ref{xpprop4} together with the cyclicity of the trace, while the last equality follows from property \ref{xpprop3}. 

To complete the calculation, we note that the product $X_{AB}^\alpha X_{BC}^\beta X_{ABC}$ is positive definite since it is a product of commuting positive definite operators (by properties \ref{xpprop1}-\ref{xpprop2}). It then follows that the above trace (\ref{tracesnet}) is strictly positive since $P_{ABC}$ is a (nonzero) projector. We conclude that
\begin{align}
\<\rho_{AB}^\alpha \rho_{BC}^\beta \> > 0
\end{align}
so that $\wab = 1$. This result establishes (\ref{wmn_conjecture}) for string-net ground states since these states all have $c_- = 0$, as discussed previously.

\section{Discussion}
\label{sec:discussion}
In this paper, we have proposed a R\'enyi generalization of the modular commutator, $\wab$, parameterized by two positive real numbers $\alpha, \beta$. Our quantity $\wab$ takes as input a gapped many-body state defined on an infinite 2D lattice, and produces as output a $U(1)$ phase. We have shown that for several classes of systems, this $U(1)$ phase $\wab$ takes a universal value related to the chiral central charge $c_-$ via the formula (\ref{wmn_conjecture}). In particular, we have derived (\ref{wmn_conjecture}) for gapped ground states of non-interacting fermion systems as well as ground states of string-net models.

As we mentioned in the introduction, the relationship (\ref{wmn_conjecture}) between $\wab$ and the chiral central charge $c_-$ does not hold in complete generality: there are (fine-tuned) counterexamples to (\ref{wmn_conjecture}) where $\wab$ takes on spurious values unrelated to the chiral central charge (see e.g.~the examples discussed in Ref.~\cite{Gass_2024}). Instead, the best we can hope for is that (\ref{wmn_conjecture}) holds for generic gapped states for a suitable definition of ``generic.'' A natural direction for future work is to understand whether this optimistic scenario holds. That is, do there exist counterexamples to (\ref{wmn_conjecture}) that are stable to small perturbations or are all counterexamples fine-tuned? These questions are also open for the original modular commutator $J$, but they may be easier to answer for its R\'enyi generalization $\wab$.

Another important question is to find a more direct connection between $\wab$ and the chiral central charge. One possible approach would be to use the alternative definition of $c_-$ recently proposed in Refs.~\cite{sopenko2024index, sopenko2025reflection} in the case of invertible states. This alternative definition is phrased in terms of permutation symmetry defects of replica systems, so it is possible that one can make a direct connection with the replica representation for $\wab$ (\ref{replicaform}).

Finally, we note that if we think of $\wab$ in terms of the replica representation (\ref{replicaform}), then $\wab$ can be regarded as part of a much larger family of quantities defined in terms of the expectation values of replica permutation operators acting on different subsystems. It seems likely that these kinds of replica permutation probes can extract more information than just the chiral central charge. In particular, Ref.~\cite{Sheffertopologicalspins} proposed a class of replica permutation probes that give information about topological spins and quantum dimensions of anyon excitations. It would be interesting to understand these probes in more generality.

\acknowledgments

J.G. and M.L. acknowledge the support of the Leinweber Institute for Theoretical Physics at the University of Chicago. This work was supported in part by a Simons Investigator grant (M.L.) and by the Simons Collaboration on Ultra-Quantum Matter, which is a grant from the Simons Foundation (651442, M.L.).

\emph{Note added} --- Related results on a R\'enyi generalization of the modular commutator were obtained independently in the recent paper, Ref.~\cite{sheffer2025}. 

\begin{appendix}

\section{Properties of quasidiagonal operators}
\label{app:quasidiagonal}
In this Appendix we derive properties (i) and (ii) of quasidiagonal operators given in \ref{sec:qd_nuP}. We restate them here for convenience:
\begin{enumerate}
    \item[(i)] If $X$ and $Y$ are quasidiagonal, their sum $X+Y$ and product $XY$ are as well.
    \item[(ii)] Let $X$ be quasidiagonal with spectrum $\sigma(X)$ and $f$ a holomorphic function on an open bounded region $D$ containing $\sigma(X)$. Then $f(X)$ is quasidiagonal. Moreover, $f(X)$ depends locally on $X$: for any quasidiagonal operator $V_R$ supported in region $R$ such that $\sigma(X+V_R)\subset D$, the matrix elements of $f(X+V_R)-f(X)$ decay exponentially outside of region $R$.
\end{enumerate}

We begin with property (i). Since $X$ and $Y$ are both quasidiagonal, there exists some $b,\gamma>0$ such that 
\begin{equation*}
    |X_{jk}|,|Y_{jk}|\leq be^{-\gamma|j-k|}.
\end{equation*}
The fact that their sum is quasidiagonal is obvious from the triangle inequality. To see that their product is quasidiagonal, choose some $0<\gamma'<\gamma$, and note that
\begin{align*}
    |(XY)_{jk}|&\leq \sum_l |X_{jl}||Y_{lk}|\\
    &\leq b^2\sum_l e^{-\gamma(|j-l|+|l-k|)}\\
    &\leq b^2 e^{-\gamma'|j-k|}\sum_l e^{-(\gamma-\gamma')(|j-l|+|l-k|)}\\
    &\leq b^2 \left(\sup_r\sum_l e^{-(\gamma-\gamma')|r-l|}\right)^2e^{-\gamma'|j-k|}.
\end{align*}

To show property (ii), we make use of a well-known fact about operators with exponential off-diagonal decay~\cite{jaffard_1990}: if $A$ is invertible with bounded inverse, then 
\begin{equation}
    |A_{jk}|\leq be^{-\gamma|j-k|}\implies |A_{jk}^{-1}|\leq b'e^{-\gamma'|j-k|},
\end{equation}
where $b',\gamma'>0$ depend only on $b,\gamma,$ and $\norm{A^{-1}}$. We then define $f(X)$ via the holomorphic functional calculus:
\begin{equation}
    f(X)=\frac{1}{2\pi i}\int_{\Gamma}f(z)\frac{1}{z-X}dz,
\end{equation}
where $\Gamma$ is a closed, rectifiable curve in $D$ that encloses $\sigma(X)$. For every $z\in\Gamma$, $(z-X)$ is quasidiagonal and invertible, so there exists $b_z,\gamma_z>0$ such that
\begin{equation*}
    \left|(z-X)^{-1}_{jk}\right|\leq b_z e^{-\gamma_z|j-k|}.
\end{equation*}
Let $b_0=\max_{z\in\Gamma}b_z$ and $\gamma_0=\min_{z\in\Gamma}\gamma_z$, then
\begin{equation*}
    |f(X)_{jk}|\leq \left(\frac{b_0}{2\pi}\int_{\Gamma}|f(z)|dz\right)e^{-\gamma_0|j-k|},
\end{equation*}
so $f(X)$ is quasidiagonal. To see the local dependence, choose $\Gamma$ so that it encloses both $\sigma(X)$ and $\sigma(X+V_R)$. Then by the second resolvent identity,
\begin{align}
    f(X+V&_R)-f(X)\nonumber\\
    &=\frac{1}{2\pi i}\int_{\Gamma}f(z)\left(\frac{1}{z-X-V_R}-\frac{1}{z-X}\right)dz\nonumber\\
    &=\frac{1}{2\pi i}\int_{\Gamma}f(z)\frac{1}{z-X-V_R}V_R\frac{1}{z-X}dz.
\end{align}
Since $(z-X-V_R)^{-1}$ and $(z-X)^{-1}$ are quasidiagonal, we find 
\begin{equation*}
    |(f(X+V_R)-f(X))_{jk}|\leq ce^{-\lambda\max\{|j-k|,\operatorname{dist}(j,R),\operatorname{dist}(k,R)\}},
\end{equation*}
for some $c,\lambda>0$.

\section{Derivation of Eq.~\eqref{mdef}}
\label{app:m_derivation}
In this Appendix, we derive Eq.~\eqref{mdef}, which we reprint below for convenience: 
\begin{equation}
\begin{gathered}
    \<\rho_{AB}^\alpha\rho_{BC}^\beta\>=\det M,\\
    M=(\id-P)+P(P_{AB}+P_{CD})^\alpha(P_{BC}+P_{AD})^\beta P.
\end{gathered}
\label{mdef_app}
\end{equation}

To begin, recall that traces of Gaussian fermionic operators $\exp(\hat{O})=\exp(\sum O_{ij}c_i^\dagger c_j)$ can be computed via determinants on the single-particle Hilbert space. For Hermitian $O$, we can diagonalize $O$ to obtain
\begin{align}
\label{gaussian_trace}
    \Tr\left(e^{\hat{O}}\right)=\Tr\left(e^{\sum O_{ij}c_i^\dagger c_j)}\right)
    &=\Tr\left(e^{\sum_{i}\epsilon_{i}\Gamma_i^\dagger \Gamma_i}\right)\nonumber\\
    &=\prod_i\left(1+e^{\epsilon_{i}}\right)\nonumber\\
    &=\det(\id+e^{O}).
\end{align}
This extends straightforwardly to non-Hermitian $O$ as well. In particular, since the product of Gaussian operators is also Gaussian, the finite temperature expectation value of a Gaussian operator is 
\begin{equation}
    \<e^{\hat{O}}\>_\beta=\frac{\Tr(e^{\hat{O}}e^{-\beta \hat{H}})}{\Tr (e^{-\beta \hat{H}})}=\frac{\det(\id+e^Oe^{-\beta h})}{\det(\id+e^{-\beta h})},
\end{equation}
where $\hat{H}$ is defined as in \eqref{complexH}. Taking the limit $\beta\rightarrow\infty$, we find that the 
ground state expectation value can be expressed in terms of the spectral projector as
\begin{equation}
\label{vev}
    \<e^{\hat{O}}\>=\det\left[(\id-P)+Pe^OP\right].
\end{equation}

Next, recall that the reduced density matrix $\rho_R$ of the many-body ground state $|\Psi\>$ is a Gaussian operator given in terms of $P$ by~\cite{Peschel_2003}
\begin{equation*}
    \rho_{R}=\frac{1}{Z_R}\exp\left(-\sum_{i,j\in R}k_{R,ij}c_i^\dagger c_j\right),\quad\  k_R=\ln\frac{\id-P_R}{P_R},
\end{equation*}
where, as in the main text, $P_R=RPR$ and $R$ denotes the spatial projection onto the corresponding region. Using this definition with Eq.~\eqref{vev}, we find that the expectation value in the numerator of Eq.~\eqref{wmn_definition} is
\begin{multline}
\label{pabpbc_expanded}
    \<\rho_{AB}^{\alpha}\rho_{BC}^{\beta}\>=\det\left[(\id-P)+ Pe^{-\alpha k_{AB}} e^{-\beta k_{BC}}P\right]\\
    \times Z_{AB}^{-\alpha}Z_{BC}^{-\beta}.
\end{multline}
The normalization constants are determined by requiring $\Tr \ \rho_{R}=1$, and are given by 
\begin{align}
\label{normalization}
    Z_R=\rdet{R}(\id+e^{-k_{R}})&=\rdet{R}\left(\id+\frac{P_R}{\id-P_R}\right)\nonumber\\
    &=\rdet{R}(\id-RPR)^{-1}\nonumber\\
    &=\det(\id-RPR)^{-1}\nonumber\\
    &=\det(\id-PRP)^{-1},
\end{align}
where $\rdet{R}$ indicates that the determinant is restricted to indices corresponding to sites in $R$, and in the last step we used the determinant property $\det(\id+XY)=\det(\id+YX)$.
We also have the identity 
\begin{align}
    Pe^{-\alpha k_{AB}}
    &=P\left[\left(\frac{P_{AB}}{\id-P_{AB}}\right)^{\alpha}+CD\right]\nonumber\\
    &=\left[\left(\frac{PABP}{\id-PABP}\right)^{\alpha}PAB+PCD\right]\nonumber\\
    &=(\id-PABP)^{-\alpha}\bigl[(PABP)^{\alpha}PAB\nonumber\\
    &\qquad\qquad\qquad\qquad+(\id-PABP)^{\alpha}PCD\bigr]\nonumber\\
    &=(\id-PABP)^{-\alpha}\bigl[(PABP)^{\alpha}PAB\nonumber\\
    &\qquad\qquad\qquad\qquad+(PCDP)^{\alpha}PCD\bigr]\nonumber\\
    &=(\id-PABP)^{-\alpha}P(P_{AB}^\alpha+P_{CD}^\alpha),
\end{align}
where $AB$ and $CD$ denote the projections onto the combined regions $A\cup B$ and $C\cup D$. In the second and fifth equalities we used the fact that $Pf(RPR)=f(PRP)R$ for any function $f$ with $f(0)=0$. The analogous identity for $e^{-\beta k_{BC}}P$ is 
\begin{equation}
\label{ekbcp_identity}
    e^{-\beta k_{BC}}P=(P_{BC}^\beta +P_{AD}^\beta)P(\id-PBCP)^{-\beta}.
\end{equation}
Inserting equations \eqref{normalization} through \eqref{ekbcp_identity} into Eq.~\eqref{pabpbc_expanded}, we obtain Eq.~\eqref{mdef_app} or equivalently Eq.~\eqref{mdef} from the main text.

\section{Invertibility of $M$}
\label{app:m_invertibility}
In this Appendix, we show that 
\begin{equation*}
    M=(\id-P)+P(P_{AB}+P_{CD})^\alpha(P_{BC}+P_{AD})^\beta P
\end{equation*}
is invertible for any spectral projector $P$ and any choice of regions $A, B, C, D$. This result implies that $\wab$ is a well-defined $U(1)$ phase for general non-interacting fermion systems.

For convenience, throughout this Appendix, we use the abbreviations $X=P_{AB}+P_{CD}$ and $Y=P_{BC}+P_{AD}$. We prove that $M$ is invertible by showing that $P X^\alpha Y^\beta P$ is bijective on the subspace $\imp$, i.e.~the image of $P$ on the single-particle Hilbert space $\mathcal{H}$. 

First, we note that
\begin{align*}
    PX^\alpha P&=(PABP)^{\alpha+1}+(PCDP)^{\alpha+1}\\
    &=(PABP)^{\alpha+1}+(P-PABP)^{\alpha+1},
\end{align*}
where $AB$ and $CD$ denote the projections onto the regions $A \cup B$ and $C \cup D$ as in the previous appendix. This gives
\begin{equation}
    2(PX^\alpha P)^2-PX^{2\alpha}P=Pf(PABP)P,
\end{equation}
where
\begin{equation*}
    f(x)=2[x^{\alpha+1}+(1-x)^{\alpha+1}]^2-[x^{2\alpha+1}+(1-x)^{2\alpha+1}].
\end{equation*}

Next we note that $PABP$ is a positive semi-definite operator with eigenvalues that lie in $[0,1]$. Also, one can check that $f(x)$ is positive for $x\in[0,1]$ with a minimum value of $1/2^{2\alpha}$. We therefore have the inequality
\begin{equation*}
    2PX^\alpha PX^\alpha P\geq PX^{2\alpha}P+\frac{1}{2^{2\alpha}}P\\
\end{equation*}
and therefore
\begin{equation}
PX^\alpha PX^\alpha P\geq PX^\alpha(\id-P)X^\alpha P+\frac{1}{2^{2\alpha}}P,
\end{equation}
where $A\geq B$ means that $A-B$ is positive semi-definite. We also have the analogous inequality for $Y^\beta$. We deduce then that for any unit vector $|v\>\in\imp$, 
\begin{align}
\label{inequalities}
\begin{split}
    \norm{PX^\alpha|v\>}^2&\geq \norm{(\id-P)X^\alpha|v\>}^2+\frac{1}{2^{2\alpha}}\\
    \norm{PY^\beta|v\>}^2&\geq \norm{(\id-P)Y^\beta|v\>}^2+\frac{1}{2^{2\beta}}. 
\end{split}
\end{align}
In particular, this tells us that $X^\alpha$ and $Y^\beta$ are injective on $\imp$ and that the principal angles between the subspaces $X^{\alpha}\imp$ and $\imp$, as well as the principal angles between the subspaces $Y^{\beta}\imp$ and $\imp$ are all strictly less than $\pi/4$. This bound on the principle angles proves the claim: it follows that $X^{\alpha}\imp$ and $Y^{\beta}\imp$ are strictly non-orthogonal to each other, in the sense that there is no $|y\>\in Y^\beta\imp$ that is in the null space of $PX^\alpha$ and there is no $|x\>\in X^\alpha\imp$ that is in the null space of $PY^\beta$.

For an explicit lower bound on the spectrum, note that Eq.~\eqref{inequalities} implies $PX^\alpha P$ and $PY^\beta P$ are both invertible on $\imp$. Therefore, for any unit vector $|w\>\in\imp$ there exists a unit vector $|v\>\in\imp$ such that $PX^\alpha|v\>$ is parallel to $PY^\beta|w\>$. We then find from the triangle and Cauchy-Schwarz inequalities
\begin{align}
\label{specbound}
    \norm{PX^\alpha Y^\beta|w\>}&\geq |\<v|X^\alpha Y^\beta|w\>|\nonumber\\
    &\geq |\<v|X^\alpha PY^\beta|w\>|\nonumber\\
    &\qquad\qquad-|\<v|X^\alpha(\id-P)Y^\beta|w\>|\nonumber\\
    &=\norm{PX^\alpha|v\>}\norm{PY^\beta|w\>}\nonumber\\
    &\qquad\qquad-|\<v|X^\alpha(\id-P)Y^\beta|w\>|\nonumber\\
    &\geq\norm{PX^\alpha|v\>}\norm{PY^\beta|w\>}\nonumber\\
    &\quad-\norm{(\id-P)X^\alpha|v\>}\norm{(\id-P)Y^\beta|w\>}\nonumber\\
    &\geq 1/2^{2(\alpha+\beta+1)}, 
\end{align}
where the final bound comes from the following inequalities, which in turn follow from Eq.~(\ref{inequalities}):
\begin{equation}
\begin{aligned}
    \norm{PX^\alpha|v\>}&\geq \norm{(\id-P)X^\alpha|v\>}+\frac{1}{2^{2\alpha+1}}\\
    \norm{PY^\beta|v\>}&\geq \norm{(\id-P)Y^\beta|v\>}+\frac{1}{2^{2\beta+1}}. 
\end{aligned}
\end{equation}

\section{Details for $\det \uabc$}
\label{app:determinant_derviation}
In this Appendix we fill in some gaps in the arguments that follow Eq.~\eqref{uabc_claim}. To begin, we note that Eq.~\eqref{uabc_traces} can be derived more rigorously by inserting a regulator $\epsilon>0$ into the definitions of the $O_i$ in Eq.~\eqref{oi_def}: 
\begin{align}
    O_1^\epsilon&=\ln\left(\ide-\rc\right)^\alpha\nonumber\\
    O_2^\epsilon&=\ln\left(\id-\ra-\rc+\frac{\rce^{\alpha+1}}{(\ide-\rc)^{\alpha}}+\frac{\rae^{\beta+1}}{(\ide-\ra)^{\beta}}\right)\\
    O_3^\epsilon&=\ln\left(\ide-\ra\right)^\beta\nonumber,
\end{align}
where $X_\epsilon\equiv X+\epsilon\id$. The $O_i^\epsilon$ are bounded operators, since $\ra$ and $\rc$ have spectra that lie in $[0,1]$. We then define $\mabc^\epsilon=e^{O_1^\epsilon}e^{O_2^\epsilon}e^{O_3^\epsilon}$ with $\uabc^\epsilon$ as its unitary part so that 
\begin{equation*}
    \det \uabc = \lim_{\epsilon\rightarrow 0} \det \uabc^\epsilon.
\end{equation*}
We apply the Baker-Campbell-Hausdorff expansion on both $\mabc^\epsilon$ and $|\mabc^\epsilon|$, which gives
\begin{align}
    \mabc^\epsilon&=e^{O_1^\epsilon}e^{O_2^\epsilon}e^{O_3^\epsilon}\nonumber\\
    &=\exp\biggl\{O_1^\epsilon+O_2^\epsilon+O_3^\epsilon\nonumber\\
    &\qquad\quad\ +\frac{[O_1^\epsilon,O_2^\epsilon]+[O_2^\epsilon,O_3^\epsilon]+[O_1^\epsilon,O_3^\epsilon]}{2}+(\cdots)\biggr\}
\end{align}
and
\begin{align}
    |\mabc^\epsilon|&=\exp\left\{\frac{1}{2}\ln(e^{O_1^\epsilon}e^{O_2^\epsilon}e^{2O_3^\epsilon}e^{O_2^\epsilon}e^{O_1^\epsilon})\right\}\nonumber\\
    &=\exp\biggl\{O_1^\epsilon+O_2^\epsilon+O_3^\epsilon+(\cdots)\biggr\},
\end{align}
where $(\cdots)$ represents higher commutators. Combining these two, we find
\begin{align}
    \uabc^\epsilon&=\mabc^\epsilon|\mabc^\epsilon|^{-1}\nonumber\\
    &=\exp\biggl\{\frac{[O_1^\epsilon,O_2^\epsilon]+[O_2^\epsilon,O_3^\epsilon]+[O_1^\epsilon,O_3^\epsilon]}{2}+(\cdots)\biggr\}.
\end{align}

We now claim that the commutators between the $O_i^\epsilon$ are ``trace class.'' Recall that an operator $T$ is trace class if $\Tr|T|<\infty$, where $|T|=\sqrt{T^\dagger T}$. Being trace class guarantees that $T$ has a well-defined trace that converges to a finite number, independent of the chosen basis. We also have the cyclic property for such operators: if $T$ is trace class and $O$ is bounded, then $\Tr(OT)=\Tr(TO)$. This cyclic property is particularly important in our case: once we show that the above commutators are trace class then by the cyclic property, it follows that all the higher commutators in the above expression for $\det \uabc^\epsilon$ have vanishing trace and therefore
\begin{multline}
    \det \uabc^\epsilon=\exp\frac{1}{2}\biggl\{\Tr[O_1^\epsilon,O_2^\epsilon]+\Tr[O_2^\epsilon,O_3^\epsilon]\\
    +\Tr[O_1^\epsilon,O_3^\epsilon])\biggr\}.
\end{multline}
We can then derive Eq.~\eqref{uabc_traces} by taking the limit $\epsilon\rightarrow 0$.

We can see that the $[O_i^\epsilon,O_j^\epsilon]$ are indeed trace class via the following lemma:
\newtheorem{theorem}{Theorem}[section]
\newtheorem{lemma}[theorem]{Lemma}
\begin{lemma}
\label{lem:trace_class}
    Let $X$ and $Y$ be bounded operators with trace class commutator. For any holomorphic functions $f$ and $g$ for which we can define $f(X)$ and $g(Y)$ via the holomorphic functional calculus, $[f(X),g(Y)]$ is also trace class.
\end{lemma}
\begin{proof}
    We have
    \begin{equation}
    \begin{aligned}
        f(X)&=\frac{1}{2\pi i}\int_{\Gamma_1}f(z)\frac{1}{z\id-X}dz,\\
        g(Y)&=\frac{1}{2\pi i}\int_{\Gamma_2}g(w)\frac{1}{w\id-Y}dw,
    \end{aligned}
    \end{equation}
    where $\Gamma_{1(2)}$ is a closed rectifiable curve in the domain of $f(g)$ that encloses the spectrum of $X(Y)$. We have the resolvent commutator identity
    \begin{multline}
        \left[\frac{1}{z\id-X},\frac{1}{w\id-Y}\right]\\=\frac{1}{z\id-X}\frac{1}{w\id-Y}[X,Y]\frac{1}{w\id-Y}\frac{1}{z\id-X},
    \end{multline}
    and it can be shown that $\Tr|AB|\leq \norm{A}\Tr|B|$. Combining these two facts, we find
    \begin{multline*}
        \Tr\left|\left[\frac{1}{z\id-X},\frac{1}{w\id-Y}\right]\right|\\\leq \norm{\frac{1}{z\id-X}}^2\norm{\frac{1}{w\id-Y}}^2\Tr\left|[X,Y]\right|.
    \end{multline*}
    We therefore have that $\Tr|[f(X),g(Y)]|\leq ab\Tr|[X,Y]|$, where $a$ and $b$ are finite constants given by
    \begin{align*}
        a&=\frac{1}{2\pi}\int_{\Gamma_1}|f(z)|\norm{\frac{1}{z\id-X}}^2 dz,\\
        b&=\frac{1}{2\pi}\int_{\Gamma_2}|g(w)|\norm{\frac{1}{w\id-Y}}^2 dw.
    \end{align*}
\end{proof}
To apply this lemma, we use the fact that the commutator $[\ra,\rc]$ is trace class, which can be established using the exponential decay of the diagonal elements of $[\ra,\rc]$ away from the triple point. The lemma then implies that  $[O_1^\epsilon,O_3^\epsilon]$ is trace class. The fact that the two commutators $[O_1^\epsilon,O_2^\epsilon]$ and $[O_2^\epsilon,O_3^\epsilon]$ are also trace class follows by applying the lemma twice. 

Having derived \eqref{uabc_traces}, all that remains is to compute the traces in this expression. To do this, we note that any product that contains at least one copy of each of the three operators $\ra,\rb,\rc$ will decay exponentially away from the triple point (which follows from quasidiagonality), and will be trace class. Furthermore, any operator obtained by cyclically permuting the terms in the product will also be trace class and have the same trace. Therefore, any commutator of the form $[L,R]$ such that the product $LR$ contains at least one copy each of $\ra,\rb,\rc$ will have vanishing trace. This gives us the following identity, for any function $g(X,Y)$ that can be written formally as a power series in $X$ and $Y$:
\begin{align}
    \Tr[\rc,g(\ra,\rc)]&=\Tr\left[\rc,g(\ra,\rc)-g(0,\rc)\right]\nonumber\\
    &=\Tr\left[\rc,g(\ra,\id-\ra-\rb)-g(0,\id-\ra-\rb)\right]\nonumber\\
    &=\Tr\left[\rc,g(\ra,\id-\ra)-g(0,\id-\ra)\right].
\end{align}
In the first line, we subtract $g(0,\rc)$ so that every term in the formal expansion of the right half of the commutator explicitly contains at least one copy of $\ra$. In the second line, we insert the identity $\id=(\id-P)+\ra+\rb+\rc$. The third line then follows from the fact that, since each term in the commutator contains at least one copy of $\ra$ from the right side and one copy of $\rc$ from the left side, any term that also contains $\rb$ will not contribute to the overall trace and can be neglected. 

Applying this identity to $\Tr[O_1,O_2]$, we find
\begin{align}   
\label{o1o2}
    \Tr[O_1,O_2]&=\Tr\left[O_1,\ln\left(\frac{(\id-\ra)^{\alpha+1}}{\ra^\alpha}+\frac{\ra^{\beta+1}}{(\id-\ra)^{\beta}}\right)\right]\nonumber\\
    &\quad-\Tr\left[O_1,\ln\left(\ra+\frac{(\id-\ra)^{\alpha+1}}{\ra^\alpha}\right)\right]\nonumber\\
    &=\Tr\left[O_1,\ln\left((\id-\ra)^{\alpha+\beta+1}+\ra^{\alpha+\beta+1}\right)\right]\nonumber\\
    &\quad-\Tr\left[O_1,\ln\left((\id-\ra)^{\alpha+1}+\ra^{\alpha+1}\right)\right]\nonumber\\
    &\quad-\Tr\left[O_1,\ln(\id-\ra)^{\beta}\right].
\end{align}
All of these terms are now in a form where we can use \eqref{integralform}, with the relevant integral being
\begin{equation*}
    \int_{0}^{1} \ln (1-t) \partial_{t} \ln \left[(1-t)^x+t^x\right] dt=-\frac{\pi^2}{12}\left(x-\frac{1}{x}\right)
\end{equation*}
for $x>0$. Explicitly, summing the contributions from the traces in \eqref{o1o2} gives
\begin{equation*}
    \Tr[O_1, O_2]
    =-\frac{\pi i\nu(P)}{24}\alpha\left(\frac{1}{\alpha+1}-\frac{1}{\alpha+\beta+1}-\beta\right),
\end{equation*}
as in the main text. The third trace $\Tr[O_2,O_3]$ can be broken up in a similar fashion:
\begin{align}   
\label{o2o3}
    \Tr[O_2,O_3]&=-\Tr[O_3,O_2]\nonumber\\
    &=-\Tr\left[O_3, \ln\left((\id-\rc)^{\alpha+\beta+1}+\rc^{\alpha+\beta+1}\right)\right]\nonumber\\
    &\quad+\Tr\left[O_3, \ln\left((\id-\rc)^{\beta+1}+\rc^{\beta+1}\right)\right]\nonumber\\
    &\quad+\Tr\left[O_3, \ln(\id-\rc)^{\alpha}\right].
\end{align}
Note that \eqref{o2o3} is equal to \eqref{o1o2} after swapping $\alpha\leftrightarrow \beta$, $\ra\leftrightarrow\rc$, and multiplying the right side by $(-1)$. This sign change compensates for the fact that we are now taking the trace of a commutator with opposite chirality of that in \eqref{o1o2}. Put differently, we have the general property
\begin{equation*}
    \Tr[f(\ra),g(\rc)]=-\Tr[f(\rc),g(\ra)],
\end{equation*}
so we find that $\Tr[O_2,O_3]$ is equal to $\Tr[O_1,O_2]$ under the exchange of $\alpha$ and $\beta$, giving 
\begin{equation*}
    \Tr[O_2, O_3]
    =-\frac{\pi i\nu(P)}{24}\beta\left(\frac{1}{\beta+1}-\frac{1}{\alpha+\beta+1}-\alpha\right),
\end{equation*}
as in the main text. 

\section{Majorana derivation}
\label{app:majorana}

In this Appendix, we derive Eq.~\eqref{mdef_majorana} as well as give an expression for $\<\rho_{AB}^\alpha\rho_{BC}^\beta\>$ for Majorana fermions in terms of Pfaffians. As stated in the main text, a non-interacting Majorana fermion Hamiltonian takes the form
\begin{equation}
    \hat{H}=\frac{i}{4}\sum_{j,k}h_{jk}\chi_j \chi_k\qquad \{\chi_j,\chi_k\}=2\delta_{jk},
\end{equation}
where $h$ is a real skew-symmetric matrix of even dimension that has eigenvalues $\pm i\epsilon_k$. Such an operator can be put into the canonical form 
\begin{equation*}
    \hat{H}=\frac{i}{2}\sum_{k}\epsilon_k b_k' b_k''=\sum_k \epsilon_k(a_k^\dagger a_k-1/2), \quad\epsilon_k \geq 0
\end{equation*}
so that the trace of its exponential is given by
\begin{align}
\label{majorana_trace}
    \Tr e^{\hat{H}}= \prod_k \left(e^{\epsilon_k/2}+e^{-\epsilon_k/2}\right)
    &=\frac{\prod_k\left(e^{\epsilon_k}-e^{-\epsilon_k}\right)}{\prod_k\left(e^{\epsilon_k/2}-e^{-\epsilon_k/2}\right)}\nonumber\\
    &=\frac{\operatorname{pf}[\sinh(ih)]}{\operatorname{pf}[\sinh(ih/2)]},
\end{align}
where $\operatorname{pf}$ is the Pfaffian. Since both sides of Eq.~\eqref{majorana_trace} are analytic in the matrix elements of $h$, this identity holds for arbitrary (i.e.~not necessarily real) skew-symmetric matrices.

Now let $\hat{O}$ be another quadratic Majorana operator defined by $\hat{O}=\frac{i}{4}\sum_{j,k}O_{jk}\chi_j\chi_k$, where $O$ is skew-symmetric. Then $e^{\hat{O}}e^{\hat{H}}=e^{\hat{K}}$, where $\hat{K}=\frac{1}{4}\sum_{j,k}\ln(e^{iO}e^{ih})_{jk}\chi_j\chi_k$, and where $\ln(e^{iO}e^{ih})$ is also skew-symmetric. We can therefore use Eq.~\eqref{majorana_trace} to find the finite-temperature expectation value of $e^{\hat{O}}$:
\begin{align}
    \<e^{\hat{O}}\>_\beta&=\frac{\Tr(e^{\hat{O}}e^{-\beta \hat{H}})}{\Tr (e^{-\beta \hat{H}})}\nonumber\\
    &=\frac{\operatorname{pf}[\sinh\ln(e^{iO}e^{-\beta ih})]}{\operatorname{pf}[\sinh\frac{1}{2}\ln(e^{iO}e^{-\beta ih})]}\frac{\operatorname{pf}[\sinh(-\beta ih/2)]}{\operatorname{pf}[\sinh(-\beta ih)]}.
\end{align}
The ground state expectation value is then given by taking the limit $\beta\rightarrow\infty$. 

As mentioned in the main text, rather than evaluate this equation further, it is easier to instead compute the square using the identity $\operatorname{pf}(A)^2=\det A$, and then determine the square root via continuity. Inserting this identity above, we find 
\begin{align}
    \<e^{\hat{O}}\>^2&=\lim_{\beta\rightarrow\infty}\frac{\det[2\cosh\frac{1}{2}\ln(e^{iO}e^{-\beta ih})]}{\det[{2\cosh(-\beta ih/2)}]}\nonumber\\
    &=\lim_{\beta\rightarrow\infty}\frac{\det[\id+e^{iO}e^{-\beta ih}]}{\det[\id+e^{-\beta ih}]}\frac{\det[e^{-\frac{1}{2}\ln(e^{iO}e^{-\beta ih})}]}{\det[e^{-\beta ih/2}]}\nonumber\\
    &=\lim_{\beta\rightarrow\infty}\frac{\det[\id+e^{iO}e^{-\beta ih}]}{\det[\id+e^{-\beta ih}]}\nonumber\\
    &=\det[(\id-P)+Pe^{iO}P],
\end{align}
where the third equality comes from the fact that the determinant of an exponential of a skew-symmetric matrix is 1. 

For non-interacting Majorana fermions, reduced density matrices of the ground state take the form
\begin{equation*}
    \rho_{R}=\frac{1}{Z_R}\exp\left(-\frac{i}{4}\sum_{j,k\in R}k_{R,jk}\chi_j \chi_k\right),\ \ ik_R=\ln\frac{\id-P_R}{P_R}.
\end{equation*}
Note that $ik_R$ is skew-symmetric since $P^T=\id-P$. We can therefore replace $e^{\hat{O}}$ with $\rho_{AB}^\alpha\rho_{BC}^\beta$ in the above equations to find expressions for $\<\rho_{AB}^\alpha\rho_{BC}^\beta\>$ and $\<\rho_{AB}^\alpha\rho_{BC}^\beta\>^2$. For the square, we find 
\begin{multline}
    \<\rho_{AB}^{\alpha}\rho_{BC}^{\beta}\>^2=\det\left[(\id-P)+ Pe^{-\alpha ik_{AB}} e^{-\beta ik_{BC}}P\right]\\
    \times Z_{AB}^{-\alpha}Z_{BC}^{-\beta},
\end{multline}
which has the same right-hand side as Eq.~\eqref{pabpbc_expanded} for complex fermions. From this point, we simply repeat the same steps following \eqref{pabpbc_expanded} to derive Eq.~\eqref{mdef_majorana}.

\end{appendix}

\bibliography{Renyi_modcom}
\end{document}